\documentclass[11pt]{article}
\usepackage[letterpaper, margin=1in]{geometry}

\usepackage{algorithm}
\usepackage{algpseudocode}
\usepackage{amsthm}
\usepackage{amsmath}
\usepackage{amssymb}
\usepackage{amsfonts}
\usepackage{bbm}
\usepackage{graphicx}
\usepackage{epstopdf}
\usepackage{dcolumn}
\usepackage{bm}
\usepackage{braket}
\usepackage{color}
\usepackage{enumerate}
\usepackage{mathtools}
\usepackage{makecell}
\usepackage{tikz-cd}
\usepackage{mathtools}
\usepackage{amssymb}
\usepackage{enumitem}
\usepackage{bm} 
\usepackage{multirow}
\usepackage{booktabs}
\usepackage{array}
\usepackage[colorlinks=true, allcolors=blue]{hyperref}
\usepackage{url}
\usepackage{authblk}

\theoremstyle{plain}\newtheorem{theorem}{Theorem}
\theoremstyle{definition}\newtheorem{remark}{Remark}
\theoremstyle{definition}
\theoremstyle{plain}\newtheorem{Pp}[theorem]{Proposition}
\theoremstyle{plain}\newtheorem{corollary}{Corollary}
\theoremstyle{plain}\newtheorem{lemma}{Lemma}
\theoremstyle{plain}
\newtheorem{defi}{Definition}
\newtheorem{thm}{Theorem}

\makeatletter
\newcommand{\ead@emailstring}{E-mail address}
\makeatother
\allowdisplaybreaks

\usepackage{amssymb}
\usepackage{amsmath}

\title{Probabilistic Bounds on the Number of Elements to Generate Finite Nilpotent Groups and Their Applications}

\author[1,2]{Ziyuan Dong\thanks{Email: dongzy5@mail2.sysu.edu.cn}}
\author[3]{Xiang Fan}
\author[3]{Tengxun Zhong}
\author[1,2]{Daowen Qiu\thanks{Email: issqdw@mail.sysu.edu.cn (corresponding author)}}

\affil[1]{Institute of Quantum Computing and Software, School of Computer Science and Engineering, Sun Yat-sen University, Guangzhou 510006, China}
\affil[2]{The Guangdong Key Laboratory of Information Security Technology, Sun Yat-sen University, Guangzhou 510006, China}
\affil[3]{School of Mathematics, Sun Yat-sen University, Guangzhou 510275, China}

\begin{document}

\maketitle



\begin{abstract}
This work establishes a new probabilistic bound on the number of elements to generate finite nilpotent groups. Let $\varphi_k(G)$ denote the probability that $k$ random elements generate a finite nilpotent group $G$. For any $0 < \epsilon < 1$, we prove that $\varphi_k(G) \ge 1 - \epsilon$ if $k \ge \operatorname{rank}(G) + \lceil \log_2(2/\epsilon) \rceil$ (a bound based on the group rank) or if $k \ge \operatorname{len}(G) + \lceil \log_2(1/\epsilon) \rceil$ (a bound based on the group chain length). Moreover, these bounds are shown to be nearly tight. Both bounds sharpen the previously known requirement of $k \ge \lceil \log_2 |G| + \log_2(1/\epsilon) \rceil + 2$. Our results provide a foundational tool for analyzing probabilistic algorithms, enabling a better estimation of the iteration count for the finite Abelian hidden subgroup problem (AHSP) standard quantum algorithm and a reduction in the circuit repetitions required by Regev's factoring algorithm.
\end{abstract}



\section{Introduction}{\label{Sec1}}

Probabilistic group theory plays a fundamental role in understanding random generation behavior in finite groups, with significant applications across computational group theory, cryptography, randomized algorithms, and quantum algorithms~\cite{detinko2013probabilistic}. For a finite group $G$, let $\varphi_k(G)$ denote the probability that $k$ independently and uniformly chosen random elements generate $G$. Formally,
\[
\varphi_k(G) = \frac{N_k(G)}{|G|^k},
\]
where $N_k(G)$ counts the number of $k$-tuples $(g_1, \ldots, g_k) \in G^k$ that generate $G$.

For any finite group $G$, Pak~\cite{pak1999probability} established that $\varphi_k(G) \ge 1-\epsilon$ whenever $k \ge \lceil \log_2 |G| + 2 + \log_2 (1/\epsilon) \rceil$ for any $0<\epsilon<1$. This was refined by Detomi and Lucchini~\cite{detomi2002many}, who proved the existence of a constant $c_\epsilon$ such that $\varphi_k(G) \ge 1-\epsilon$ whenever $k \ge \lfloor \mathrm{rank}(G) + c_\epsilon \log_2 \log_2 |G| \log_2 \log_2 \log_2 |G| \rfloor$, where $\mathrm{rank}(G)$ denotes the minimal number of generators; however, $c_\epsilon$ lacks explicit expression. For more structured classes, stronger results emerge. For solvable groups $G$, Pak showed $\varphi_k(G) \ge 1/e \approx 0.367879\ldots$ when $k \ge 3.25\,\mathrm{rank}(G) + 10^7$~\cite{pak1999probability}. Further specializing to nilpotent groups $G$—which form a proper subclass of solvable groups and include all finite Abelian groups—the bound improves: $\varphi_k(G) \ge 1/e$ whenever $k \ge \mathrm{rank}(G) + 1$~\cite{pak1999probability}.

Beyond probability bounds, the expected number $\mathcal{E}(G)$ of random elements required to generate a group has been extensively studied. Pomerance~\cite{pomerance2002expected} proved that for finite Abelian groups (and more generally, nilpotent groups), $\mathcal{E}(G) \le \mathrm{rank}(G) + \sigma$, where $\sigma = 2.118456\ldots$ is a constant given in terms of the Riemann zeta function. Furthermore, for finite nilpotent groups, the generation probability by $\mathrm{rank}(G) + 1$ random elements exceeds $c = 0.435757\ldots$, improving the $1/e$ bound in~\cite{pak1999probability}. For general finite groups, Lubotzky~\cite{lubotzky2002expected} established $\mathcal{E}(G) \le e \cdot \mathrm{rank}(G) + 2e \log_2 \log_2 |G| + 11$, providing the first general upper bound for $\mathcal{E}(G)$. Lubotzky also showed $\varphi_k(G) \ge 1/e$ whenever $k \ge \mathrm{rank}(G) + 2\log_2 \log_2 |G| + 4.02$, confirming Pak's conjecture. More recently, Lucchini~\cite{lucchini2016expected} derived an exact formula for $\mathcal{E}(G)$ via the M\"obius function of the subgroup lattice, yielding precise values for nontrivial finite groups.

Despite these advances, existing probabilistic analyses remain insufficient for algorithmic applications. We bridge this gap in Theorem~\ref{zhihuolianchangdingli} by proving that for any nilpotent group $G$ and $0<\epsilon<1$, $\varphi_k(G) \ge 1-\epsilon$ whenever 
\[
k \ge \mathrm{rank}(G) + \lceil \log_2 (2/\epsilon) \rceil
\]
or
\[
k \ge \mathrm{len}(G) + \lceil \log_2 (1/\epsilon) \rceil.
\]
Both bounds sharpen the previously known requirement $k \ge \lceil \log_2 |G| + \log_2 (1/\epsilon) \rceil + 2$. We also show in Theorem~\ref{zhihuolianchangdinglijinxing} that these bounds are nearly tight. Specifically, via a constructive counterexample, it shows that the group chain length bound can be reduced by at most 1, and the group rank bound by at most 2.

Our theorem enables diverse applications, including two major examples in quantum computing. The hidden subgroup problem (HSP)~\cite{kitaev1995quantum} constitutes a central challenge in quantum computation, encompassing most known exponential speedups achieved via quantum Fourier transform~\cite{nielsen_quantum_2010}. Each iteration of the standard quantum algorithm for the finite Abelian HSP (AHSP)~\cite{kitaev1995quantum,nielsen_quantum_2010} yields an element from $H^\perp$ (see Definition~\ref{zhengjiaoqun}), where $H^\perp \leq G$ denotes the dual orthogonal subgroup of the hidden subgroup $H \leq G$. Applying our probabilistic group theory result, we determine the iteration count for achieving success probability $1-\epsilon$ in the standard quantum finite AHSP algorithm to be either $\mathrm{rank}(G) + \lceil \log_2 (2/\epsilon) \rceil$ or $\mathrm{len}(G) - \mathrm{len}(H) + \lceil \log_2 (1/\epsilon) \rceil$. The former achieves an exponential improvement in $\epsilon$-dependence over the prior bound $\lfloor 4/\epsilon \rfloor \mathrm{rank}(G)$~\cite{koiran2007quantum}, while the latter improves upon $\lceil \log_2 |G| + \log_2 (1/\epsilon) + 2 \rceil$~\cite{lomont2004hidden}.

Additionally, our result optimizes iteration counts in Regev's recent quantum factoring algorithm~\cite{regev2025efficient}. While Corollary~4.2 in Regev's work states that $\mathrm{rank}(G) + 4$ uniformly random elements generate $G$ with probability at least $1/2$, we demonstrate that only $\mathrm{rank}(G) + 2$ elements suffice for the same success probability. This refinement directly reduces the quantum circuit repetition count from $\sqrt{n} + 4$ to $\sqrt{n} + 2$ in Regev's factoring algorithm while maintaining identical success probability.

The remainder of this paper is organized as follows. Section~\ref{Sec2} presents preparatory lemmas and background knowledge, including properties of $\mathrm{rank}(G)$ and $\mathrm{len}(G)$. Section~\ref{Sec3} establishes our main sampling theorem for finite nilpotent groups, analyzing the bound from the dual perspectives of group rank and chain length, and proving that the bounds are nearly tight. Section~\ref{Sec4} applies our main theorem to quantum computing: determining iteration counts for the standard quantum finite AHSP algorithm and reducing quantum circuit repetitions in Regev's factoring algorithm. Section~\ref{Sec5} summarizes our main conclusions.

\section{Preliminaries and Earlier Results}\label{Sec2}

\subsection{Earlier Results for Generating Groups}

In this subsection, we review prior results on group generation. To make the paper self-contained, we include the proofs of some lemmas.

$\mathrm{rank}(G)$ is defined as the  minimal cardinality of a generating set of $G$. And \textbf{chain length}  $\mathrm{len}(G)$ is defined as follows:

\begin{defi}[\cite{hungerford2012algebra}]\label{lianchangdingyi}
A \textbf{composition series} of a group $G$ is a maximal chain of normal subgroups
\[
\{e\} = G_0 \trianglelefteq G_1 \trianglelefteq \cdots \trianglelefteq G_n = G
\]
such that each quotient $G_i/G_{i-1}$ is simple. The integer $n$ is called the \textbf{chain length} of $G$, denoted $\mathrm{len}(G)$.
\end{defi}

\begin{remark}
For any given finite group $G$, the \textbf{chain length} $\mathrm{len}(G)$ is an invariant, though \textbf{composition series} may not be unique. 
\end{remark}

\begin{defi}[\textbf{Frattini Subgroup}~\cite{isaacs2008finite}]\label{fuladiniziqundingyi}
    For any group $ G $, the \textbf{Frattini subgroup} $ \Phi(G) $ is defined as the intersection of all maximal subgroups:

\[
\Phi(G) := \bigcap_{H \in \mathcal{H}} H \quad \text{where} \quad \mathcal{H} = \{ H \leq G \mid H \text{ is a maximal subgroup} \}.
\]
   
   We conventionally set $\Phi(G) = G$ if $G$ has no maximal subgroups.
\end{defi}

\begin{remark}\label{fuladiniziqundengjiadingyi}
$\Phi(G)$ is always a normal subgroup of $G$.

Definition~\ref{fuladiniziqundingyi} is equivalent to defining $\Phi(G)$ as the set of all \textbf{non-generators} of $ G $. Formally,
    $$
    \Phi(G) := \left\{ g \in G \mid \forall S \subseteq G,\ \text{if } \langle S, g \rangle = G \text{ then } \langle S \rangle = G \right\}
    $$
    This implies that elements in $ \Phi(G) $ are redundant for generating $ G $. 

\end{remark}

Then we have Lemma~\ref{fuladiniziqunshengchengyinli}, which is fundamental to Lemma~\ref{fuladiniziqungailv}.

\begin{lemma}[\cite{dixon2007problems}]\label{fuladiniziqunshengchengyinli}
Let $\Phi(G)$ be the Frattini subgroup of a finite group $G$. Then a set $\{g_1, \dots, g_k\} \subseteq G$ generates $G$ if and only if its projection $\{\bar{g}_1,\dots,\bar{g}_k\}$ generates the quotient group $G/\Phi(G)$, where $\bar{g}_i = g_i\Phi(G)$.
\end{lemma}

\begin{proof}

($\Longrightarrow$) If $\langle g_1, \dots, g_k \rangle = G$, then applying the quotient homomorphism gives $\langle \bar{g}_1, \dots, \bar{g}_k \rangle = G/\Phi(G)$.

($\Longleftarrow$) Suppose $\langle \bar{g}_1, \dots, \bar{g}_k \rangle = G/\Phi(G)$. Let $H = \langle g_1, \dots, g_k \rangle$. 

For any $g \in G$, its coset $\bar{g} = g\Phi(G)$ can be written as $\bar{g} = w(\bar{g}_1, \dots, \bar{g}_k)$ for some word $w$. By the homomorphism property, this implies $g\Phi(G) = w(g_1, \dots, g_k)\Phi(G)$, so $g \in H\Phi(G)$. 

Thus $G = H\Phi(G)$. By Remark~\ref{fuladiniziqundengjiadingyi}, $\Phi(G)$ consists of non-generators, so $H = G$.
\end{proof}

\begin{lemma}[\cite{acciaro1996probability}]\label{fuladiniziqungailv}
Let $G$ be a finite group with Frattini subgroup $\Phi(G)$. Then
$$
\varphi_k(G) = \varphi_k(G/\Phi(G)).
$$
\end{lemma}
\begin{proof}
Let $N$ be the number of $k$-tuples generating $G/\Phi(G)$. 

By Lemma~\ref{fuladiniziqunshengchengyinli}, a $k$-tuple $(g_1,\dots,g_k)$ generates $G$ if and only if its projection $(\bar{g}_1,\dots,\bar{g}_k)$ generates $G/\Phi(G)$.

For each generating $k$-tuple $(\bar{g}_1,\dots,\bar{g}_k)$ of $G/\Phi(G)$, there are exactly $|\Phi(G)|^k$ preimages in $G$ (since we can multiply each $g_i$ by any element of $\Phi(G)$). Therefore, the number of generating $k$-tuples of $G$ is $N \cdot |\Phi(G)|^k$.

Hence,
$$
\varphi_k(G) = \frac{N \cdot |\Phi(G)|^k}{|G|^k} = \frac{N}{|G/\Phi(G)|^k} = \varphi_k(G/\Phi(G)),
$$
where we use $|G| = |\Phi(G)| \cdot |G/\Phi(G)|$.
\end{proof}

\begin{lemma}[\textbf{Probabilities of Generating Finite $p$-Groups}~\cite{dixon2007problems,acciaro1996probability}]\label{p-qungailv}
Let $G$ be a finite $p$-group with $\mathrm{rank}(G) = r$. Then the following holds:

\begin{enumerate}[label=\textup{(\roman*)}, itemsep=0pt,align=left,labelwidth=1.5em]
    \item $G/\Phi(G) \cong (\mathbb{Z}_p)^r$;
    \item $\varphi_k(G)= \varphi_k(\mathbb{Z}_p^{r}) = 
\begin{cases}
\displaystyle
\prod_{i=0}^{r-1} \left(1 - p^{i - k}\right), & \text{if } k \geq r,\\
0, & \text{otherwise}.
\end{cases}$
\end{enumerate}
\end{lemma}

\begin{proof}

(i) Since $G$ is a finite $p$-group, $\Phi(G) = G^p[G,G]$ and the quotient $G/\Phi(G)$ is an \textit{elementary Abelian $p$-group}~\cite{dixon2007problems}. $\mathrm{rank}(G) = \mathrm{rank}(G/\Phi(G)) = r$, so $G/\Phi(G) \cong (\mathbb{Z}_p)^r$.

(ii) $(\mathbb{Z}_p)^r$ is an $r$-dimensional linear space over the finite field $\mathbb{Z}_p$. 
Each uniform sampling generates a row vector in $(\mathbb{Z}_p)^r$, and performing $k$ independent 
uniform samplings yields a matrix $M \in \mathbb{Z}_p^{k \times r}$. 

Our objective is to ensure $M$ has row rank $r$. By the fundamental rank equality (row rank = column rank), 
this implies that $M$ has full column rank $r$ (i.e., its columns are linearly independent). Therefore, when $k \geq r$,

\[
\varphi_k(\mathbb{Z}_p^{r})=\dfrac{(p^k-1)(p^k-p)(p^k-p^2)\cdots(p^k-p^{r-1})}{(p^k)^r}=\prod_{i=0}^{r-1} \left(1 - p^{i - k}\right).
\]

Otherwise when $k<r$,
\[
\varphi_k(\mathbb{Z}_p^{r})=0.
\]

\end{proof}

\begin{corollary}\label{Z2shenglianzuinan}
$\varphi_k(\mathbb{Z}_p^{r})\ge \varphi_k(\mathbb{Z}_2^{r})$.
\end{corollary}

\begin{proof}
Clear from Lemma~\ref{p-qungailv}(ii). 
\begin{itemize}
    \item When $k\ge r$, $\varphi_k(\mathbb{Z}_p^{r}) = \displaystyle\prod_{i=0}^{r-1} \left(1 - p^{i - k}\right) \ge \displaystyle\prod_{i=0}^{r-1} \left(1 - 2^{i - k}\right) = \varphi_k(\mathbb{Z}_2^{r}).$
 
    \item When $k<r$, $\varphi_k(\mathbb{Z}_p^{r})=0=\varphi_k(\mathbb{Z}_2^{r}).$
\end{itemize}

Thus $\varphi_k(\mathbb{Z}_p^{r})\ge \varphi_k(\mathbb{Z}_2^{r})$.
\end{proof}

\begin{defi}[\textbf{Nilpotent Group}~\cite{isaacs2008finite}]
A group \(G\) is called \textbf{nilpotent} if its lower central series terminates at the trivial subgroup after finitely many steps. That is, defining \(\gamma_1(G) = G\) and \(\gamma_{i+1}(G) = [\gamma_i(G), G]\), there exists a positive integer \(c\) such that \(\gamma_{c+1}(G) = \{e\}\). The smallest such \(c\) is called the nilpotency class of \(G\). In particular, a finite nilpotent group can be decomposed as the direct product of its Sylow subgroups.
\end{defi}

\begin{lemma}[\textbf{Probability Product Formula for Finite Nilpotent Groups}~\cite{acciaro1996probability}]\label{husuqungailv}

Let $G \cong G_{p_1} \times G_{p_2} \times \cdots \times G_{p_m}$ be a finite nilpotent group, where $G_{p_i}$ are the Sylow $ p_i $-subgroups for distinct primes $p_1, \ldots, p_m$. Then for any positive integer $k$:
\[
\varphi_k(G) = \prod_{i=1}^{m} \varphi_k(G_{p_i}).
\]

\end{lemma}

\begin{proof}
See~\cite[Corollary 4]{acciaro1996probability}.
\end{proof}

\subsection{Properties of $\mathrm{rank}(G)$ and $\mathrm{len}(G)$}

This subsection establishes key properties of \(\mathrm{rank}(G)\) and \(\mathrm{len}(G)\) to support our main results in Sec.~\ref{Sec3} and the discussion on the quantum query complexity of finite AHSP algorithm in Sec.~\ref{Sec4}.

\begin{lemma}[\textbf{Additivity of Chain Length for Finite Groups}~\cite{hungerford2012algebra}]\label{thm:normal-subgroup-length}
For any finite group $G$ and normal subgroup \( H \trianglelefteq G \):
\begin{equation*}
\mathrm{len}(G) = \mathrm{len}(H) + \mathrm{len}(G/H).
\end{equation*}
\end{lemma}

\begin{lemma}[\textbf{Subadditivity of Rank for Finite Groups}]\label{subadditivity-finite}
For a finite group $G$ and normal subgroup $H \trianglelefteq G$:
\begin{equation*}
\mathrm{rank}(G) \le \mathrm{rank}(H) + \mathrm{rank}(G/H).
\end{equation*}
\end{lemma}

\begin{proof}
We establish the subadditivity of rank through generating set construction.

Let $S = \{s_1, \dots, s_m\}$ be a minimal generating set for $H$, so $\mathrm{rank}(H) = m$. Similarly, let $\{\bar{t}_1, \dots, \bar{t}_n\}$ be a minimal generating set for $G/H$ with $\mathrm{rank}(G/H) = n$. Choose representatives $t_i \in G$ such that $\bar{t}_i = t_iH$, and define $T = \{t_1, \dots, t_n\}$.

For any \( x \in G \), its coset \( xH \in G/H \) lies in the subgroup generated by \( \{\bar{t}_1, \dots, \bar{t}_n\} \). Thus, there exists an element \( g \) in the subgroup generated by \( T \) such that \( xH = gH \). This implies \( x = gk \) for some \( k \in H \).

Since \( k \in H \), it lies in the subgroup generated by \( S \). Therefore, \( x \) lies in the subgroup generated by \( S \cup T \).

This shows that \( S \cup T \) generates \( G \), though it may not be minimal. We conclude that
\[
\mathrm{rank}(G) \leq |S \cup T| \leq |S| + |T| = \mathrm{rank}(H) + \mathrm{rank}(G/H).
\]

\end{proof}

\begin{defi}[\textbf{Solvable Group}~\cite{isaacs2008finite}]
A group \(G\) is called \textbf{solvable} if it has a subnormal series
\[
\{e\} = G_0 \trianglelefteq G_1 \trianglelefteq \cdots \trianglelefteq G_n = G
\]
such that each quotient \(G_i / G_{i-1}\) is Abelian. For finite groups, this is equivalent to the condition that all composition factors are cyclic groups of prime order.
\end{defi}

\begin{remark}

For a solvable group $G$, if $|G| = \prod\limits_{i=1}^k p_i^{a_i}$ is the prime factorization, then:
\[
\mathrm{len}(G) = \sum_{i=1}^k a_i.
\]

\end{remark}

\begin{thm}\label{zhihelianchangguanxidingli}
Let $G$ be a finite solvable group. Then $\mathrm{len}(G) \geq \mathrm{rank}(G)$. Moreover, if $G \cong (\mathbb{Z}_p)^k$ for some prime $p$ and integer $k$, then $\mathrm{len}(G) = \mathrm{rank}(G)$.
\end{thm}

\begin{proof}
Since $G$ is finite and solvable, it admits a composition series:
\[
\{e\} = G_0 \trianglelefteq G_1 \trianglelefteq \cdots \trianglelefteq G_n = G
\]
with $G_i/G_{i-1}$ simple for each $i$. By the solvability of $G$, each composition factor satisfies $G_i/G_{i-1} \cong \mathbb{Z}_{p_i}$ for some prime $p_i$. The length of this series is $\mathrm{len}(G) = n$.

Applying the subadditivity of group rank (Lemma~\ref{subadditivity-finite}) iteratively along the composition series yields:
\begin{align*}
\mathrm{rank}(G) &\le \mathrm{rank}(G_0) + \mathrm{rank}(G_1/G_0) + \cdots + \mathrm{rank}(G/G_{n-1}) \\
&= 0 + 1 + \cdots + 1 \\
&= n = \mathrm{len}(G),
\end{align*}
where each $\mathrm{rank}(G_i/G_{i-1}) = 1$ since $G_i/G_{i-1} \cong \mathbb{Z}_{p_i}$.

If $G \cong (\mathbb{Z}_p)^k$, then clearly $\mathrm{len}(G) = \mathrm{rank}(G) = k$.
\end{proof}

\begin{thm}\label{lemma-rank-nilpotent-direct-two}
Let $G$ be a finite nilpotent group with Sylow decomposition
\[
G \cong G_p \times G_q,
\]
where $p$ and $q$ are distinct primes. Then the rank of $G$ satisfies
\[
\mathrm{rank}(G) = \max\{ \mathrm{rank}(G_p), \mathrm{rank}(G_q) \}.
\]
\end{thm}

\begin{proof}
Let $r = \mathrm{rank}(G)$, $r_p = \mathrm{rank}(G_p)$, and $r_q = \mathrm{rank}(G_q)$.

($\max\{r_p, r_q\} \le r$) Any generating set of $G$ with $r$ elements projects to generating sets of $G_p$ and $G_q$ under the natural projections $\pi_p: G \to G_p$ and $\pi_q: G \to G_q$, so $r_p \le r$ and $r_q \le r$. Hence, $\max\{r_p, r_q\} \le r$.

($r \le \max\{r_p, r_q\}$) Let $m = \max\{r_p, r_q\}$. Choose a minimal generating set $\{g_{p1}, \ldots, g_{pr_p}\}$ of $G_p$ and a minimal generating set $\{g_{q1}, \ldots, g_{qr_q}\}$ of $G_q$. We extend these generating sets to size $m$ by adding the identity elements $e_p$ of $G_p$ and $e_q$ of $G_q$ if $r_p < m$ or $r_q < m$, respectively.

Now define elements $x_1, \ldots, x_m \in G$ by:
\[
x_j = (g_{pj}, g_{qj}), \quad \text{where } g_{pj} = e_p \text{ if } j > r_p, \text{ and } g_{qj} = e_q \text{ if } j > r_q.
\]

We claim that $\{x_1, \ldots, x_m\}$ generates $G$. To prove this, take an arbitrary element $g = (g_p, g_q) \in G$ with $g_p \in G_p$ and $g_q \in G_q$. Since $\{g_{p1}, \ldots, g_{pr_p}\}$ generates $G_p$ and $\{g_{q1}, \ldots, g_{qr_q}\}$ generates $G_q$, we can write:
\[
g_p = w_p(g_{p1}, \ldots, g_{pr_p}), \quad g_q = w_q(g_{q1}, \ldots, g_{qr_q})
\]
for some group words $w_p$ and $w_q$.

Since $|G_p|$ and $|G_q|$ are coprime, by the Chinese Remainder Theorem, there exist integers $k_p$ and $k_q$ such that:
\[
k_p \equiv 
\begin{cases} 
1 \pmod{|G_p|} \\
0 \pmod{|G_q|}
\end{cases}, \quad
k_q \equiv 
\begin{cases} 
0 \pmod{|G_p|} \\
1 \pmod{|G_q|}
\end{cases}.
\]

Define $h_p, h_q \in G$ by:
\[
h_p = \left[w_p(x_1, \ldots, x_{r_p})\right]^{k_p}, \quad
h_q = \left[w_q(x_1, \ldots, x_{r_q})\right]^{k_q}.
\]

Then the projections satisfy:
\begin{align*}
\pi_p(h_p) &= g_p^{k_p} = g_p \quad \text{(since $k_p \equiv 1 \pmod{|G_p|}$)} \\
\pi_q(h_p) &= e_q \quad \text{(since $k_p \equiv 0 \pmod{|G_q|}$)} \\
\pi_p(h_q) &= e_p \quad \text{(since $k_q \equiv 0 \pmod{|G_p|}$)} \\
\pi_q(h_q) &= g_q^{k_q} = g_q \quad \text{(since $k_q \equiv 1 \pmod{|G_q|}$)}
\end{align*}

Thus $h_p = (g_p, e_q)$ and $h_q = (e_p, g_q)$, and therefore:
\[
g = h_p \cdot h_q \in \langle x_1, \ldots, x_m \rangle.
\]

Hence, $\{x_1, \ldots, x_m\}$ generates $G$, so $r \le m = \max\{r_p, r_q\}$.

Combining both inequalities, we conclude that $\mathrm{rank}(G) = \max\{ \mathrm{rank}(G_p), \mathrm{rank}(G_q) \}$.
\end{proof}

\section{Main Results}\label{Sec3}

Synthesizing the results in Sec~\ref{Sec2}, we establish Theorem~\ref{zhihuolianchangdingli}, which is a new result in probabilistic group theory using dual perspectives: group rank and chain length. Our approach minimally employs inequalities only for estimation, yielding a sufficient condition. Moreover, Theorem~\ref{zhihuolianchangdinglijinxing} provides a counterexample showing that the chain length bound can be reduced by at most 1, and the rank bound by at most 2.

\begin{thm}[\textbf{Bounds for Generating Finite Nilpotent Groups}]\label{zhihuolianchangdingli}

For any finite nilpotent group $G$ and $0<\epsilon<1$. Let $\varphi_k(G)$ denote the probability that $k$ elements $\{g_1,\ldots,g_k\}$, sampled uniformly and independently with replacement from $G$, generate the entire group $G$.

Then the following bounds hold:
\begin{enumerate}[label=\textup{(\roman*)}, itemsep=0pt,align=left,labelwidth=1.5em]
\item $\varphi_k(G) \ge 1-\epsilon$, if $k \ge \operatorname{rank}(G) + \left\lceil \log_2 \dfrac{2}{\epsilon} \right\rceil$.
\item $\varphi_k(G) \ge 1-\epsilon$, if $k \ge \operatorname{len}(G) + \left\lceil \log_2 \dfrac{1}{\epsilon} \right\rceil$.
\end{enumerate}

\end{thm}

\begin{proof}

Let $G$ be a finite nilpotent group with $\mathrm{rank}(G) = r$. Then $G$ admits a decomposition into its Sylow subgroups:
\[
G \cong \prod_{i=1}^h G_{p_i},
\]
where $p_1, \ldots, p_h$ are the distinct prime divisors of $|G|$, and each Sylow $p_i$-subgroup $G_{p_i}$ is a finite $p_i$-group. Let $\mathrm{rank}(G_{p_i}) = r_i$ for $i = 1, \ldots, h$; then by iterated application of Theorem~\ref{lemma-rank-nilpotent-direct-two}, it follows that the rank of $G$ satisfies $\max\limits_{1 \le i \le h} r_i = r$.

Furthermore, the chain length of $G$ equals the sum of the chain lengths of its Sylow subgroups:
\[
\mathrm{len}(G) = \sum_{i=1}^{h} \mathrm{len}(G_{p_i}).
\]

Then for all $0<\epsilon<1$, we have

(i)
\begin{align}
\varphi_k(G)
&=\prod_{i=1}^{h} \varphi_k(G_{p_i})\nonumber\\
&=\prod_{i=1}^{h} \varphi_k\Big( (\mathbb{Z}_{p_i})^{r_i} \Big)\quad\quad \text{(by Lemma~\ref{p-qungailv}(ii))}\label{shizi66}\\
&\geq \prod_{i=1}^{h} \varphi_k\Big( (\mathbb{Z}_{p_i})^r \Big)\quad \quad \text{(since $r \geq r_i$)}\label{shizi67}\\
&=\prod_{i=1}^{h} \prod_{j=0}^{r-1} \left(1 - \dfrac{1}{{p_i}^{k-j}} \right)\quad \quad \text{(by Lemma~\ref{p-qungailv}(ii))}\label{shizi68}\\
&=\prod_{i=1}^{h}\prod_{j=k-r+1}^k \left(1 - \dfrac{1}{{p_i}^{j}} \right)\quad \quad \text{(reindexing)} \nonumber\\
&\geq \prod_{i=1}^{h}\left(1 - \sum_{j=k-r+1}^{k} \dfrac{1}{{p_i}^{j}} \right) \quad\quad \text{(Bernoulli's inequality)}\nonumber \\
&=\prod_{i=1}^{h} \left(1 - \frac{1}{{p_i}^{k-r}} \cdot \frac{1 - p_i^{-r}}{p_i - 1} \right)\nonumber\\
&\geq\prod_{i=1}^{h}\left(1 - \frac{1}{{p_i}^{k-r}(p_i - 1)} \right) \quad\quad \text{(since $1 - p_i^{-r} < 1$)}\nonumber\\
&\geq 1 - \sum_{i=1}^{h} \frac{1}{{p_i}^{k-r}(p_i - 1)}.\label{shizi73} \quad\quad \text{(Bernoulli's inequality)}
\end{align}

However, Eq.~\eqref{shizi73} is still difficult to handle directly. Therefore, we further estimate it:

\begin{align}
1-\sum_{i=1}^{h}\frac{1}{{p_i}^{k-r}(p_i-1)}&\ge 1-\sum_{p\ \text{prime}}\frac{1}{{p}^{k-r}(p-1)}\nonumber\\
&\ge 1-\sum_{n=2}^{\infty}\frac{1}{{n}^{k-r}(n-1)}.\nonumber
\end{align}

To ensure the inequality $\varphi_k(G) \ge 1-\epsilon$, it suffices to bound:
\[
1-\sum_{n=2}^{\infty}\frac{1}{{n}^{k-r}(n-1)}\ge 1-\epsilon\Longleftrightarrow\frac{1}{2^{k-r}} + \sum_{n=3}^{\infty}\frac{1}{n^{k-r}(n-1)} \le \epsilon.
\]

We establish an upper bound for the series via integral estimates. Observing that
\[
\sum_{n=3}^{\infty}\frac{1}{n^{k-r}(n-1)}\le 
\int_{2}^{\infty} \frac{1}{x^{k-r}(x-1)}dx,
\]

it suffices to require:

\[
\frac{1}{2^{k-r}} + \int_{2}^{+\infty} \frac{1}{x^{k-r}(x-1)}  dx \leq \epsilon.
\]

Proceeding step by step:

\[
\begin{aligned}
&\frac{1}{2^{k-r}} + \int_{2}^{+\infty} \frac{1}{x^{k-r}(x-1)}  dx\\
=&\frac{1}{2^{k-r}} + \int_{2}^{+\infty} \sum_{n=1}^{\infty} \frac{1}{x^{k-r+n}}  dx \quad 
(\text{since } \frac{1}{x-1} = \sum_{n=1}^{\infty} \frac{1}{x^n} \text{ for } x > 1) \\
=&\frac{1}{2^{k-r}} + \sum_{n=1}^{\infty} \int_{2}^{+\infty} \frac{1}{x^{k-r+n}}  dx \quad 
\text{(uniform convergence permits interchange)} \\
=&\frac{1}{2^{k-r}} + \sum_{n=1}^{\infty} \frac{2^{-(k-r+n-1)}}{k-r+n-1} \\
=&\frac{1}{2^{k-r}} + \sum_{i=k-r}^{\infty} \frac{1}{i \cdot 2^i} \quad 
\text{(letting } i = k - r + n - 1\text{)}.
\end{aligned}
\]

Now, observe that

\[
\begin{aligned}
\sum_{i=k-r}^{\infty} \frac{1}{i \cdot 2^i} 
&\leq \frac{1}{k-r} \sum_{i=k-r}^{\infty} \frac{1}{2^i} 
= \frac{1}{k-r} \cdot \frac{1/2^{k-r}}{1 - 1/2} 
= \frac{2}{(k-r) \cdot 2^{k-r}} \\
&\leq \frac{1}{2^{k-r}}, \quad \text{for } k - r \geq 2.
\end{aligned}
\]

Therefore, it is enough to require (with \( k - r \geq 2 \)):

\[
\frac{1}{2^{k-r}} + \frac{1}{2^{k-r}} \leq \epsilon 
\quad \Longleftrightarrow \quad  k - r \geq \log_2 \frac{2}{\epsilon}.
\]

Hence, a sufficient condition is

\[
k \geq r + \left\lceil \log_2 \frac{2}{\epsilon} \right\rceil, \quad \text{for } \epsilon \in (0, 1).
\]

(ii)\begin{align}
\varphi_k(G)&=\prod_{i=1}^{h} \varphi_k(G_{p_i})\nonumber\\
&=\prod_{i=1}^{h}\varphi_k\Big((\mathbb{Z}_{p_i})^{r_i}\Big)\quad\quad \text{(by Lemma~\ref{p-qungailv}(ii))}\nonumber \\
&\ge \prod_{i=1}^{h}\varphi_k\Big((\mathbb{Z}_{p_i})^{\mathrm{len}(G_{p_i})}\Big)\quad\quad\text{($\mathrm{len}(G_{p_i}) \ge \mathrm{rank}(G_{p_i}) = r_i$ by Theorem~\ref{zhihelianchangguanxidingli})}\nonumber\\
&\ge\prod_{i=1}^{h}\varphi_k\Big((\mathbb{Z}_2)^{\mathrm{len}(G_{p_i})}\Big)\quad\quad \text{(by Corollary~\ref{Z2shenglianzuinan})}\nonumber\\
&=\prod_{i=1}^{h} \prod_{l=0}^{\mathrm{len}(G_{p_i})-1}(1-\dfrac{1}{2^{k-l}}).\quad\quad \text{(by Lemma~\ref{p-qungailv}(ii))}\nonumber
\end{align}

To continue the estimation, we define \( L_i = \sum\limits_{j=1}^{i} \mathrm{len}(G_{p_j}) \) for \( i = 0, 1, \dots, h \), with \( L_0 = 0 \), representing the cumulative length after including the first \(i\) subgroups. The total length of $h$ subgroups is then \( L_h = \sum\limits_{j=1}^{h} \mathrm{len}(G_{p_j}) \). We then have:

\begin{align}
\prod_{i=1}^{h} \prod_{l=0}^{\mathrm{len}(G_{p_i})-1}\left(1-\dfrac{1}{2^{k-l}}\right)
&\ge \prod_{i=1}^{h} 
   \prod_{l=L_{i-1}}^{L_i-1}\left(1-\dfrac{1}{2^{k-l}}\right)\label{shizi86}\\
&= \prod_{l=0}^{L_h-1}\left(1-\dfrac{1}{2^{k-l}}\right)\label{shizi87}\\
&=\prod_{l=0}^{\mathrm{len}(G)-1} \left(1-\frac{1}{2^{k-l}}\right)
\quad\text{(since $\mathrm{len}(G)= \sum\limits_{i=1}^{h}\mathrm{len}(G_{p_i})=L_h$)}\\
&=\varphi_k\left((\mathbb{Z}_{2})^{\mathrm{len}(G)}\right)
\quad\text{(by Lemma~\ref{p-qungailv}(ii))}\label{shizi89}\\
&\ge 1-\sum_{l=0}^{\mathrm{len}(G)-1}\dfrac{1}{2^{k-l}}
\quad\text{(by Bernoulli's inequality)}\label{shizi90}
\end{align}

To ensure $\varphi_k(G) \ge 1-\epsilon$, it suffices by Eq.~(\ref{shizi90}) to satisfy:

\begin{align*}
&\sum_{\mathclap{l=0}}^{\mathclap{\mathrm{len}(G)-1}}\dfrac{1}{2^{k-l}}\le \epsilon\quad\Longleftrightarrow\quad \dfrac{2^{\mathrm{len}(G)}-1}{2^k}\le \epsilon \quad\Longleftrightarrow \quad 2^k\ge \dfrac {2^{\mathrm{len}(G)}-1}{\epsilon}.
\end{align*}

This condition is satisfied when:
\[
2^k \ge \frac{2^{\mathrm{len}(G)}}{\epsilon}
\quad \Longleftrightarrow \quad 
k \ge \mathrm{len}(G) + \log_2 \frac{1}{\epsilon}.
\]

Hence, a sufficient condition is
\[
k\ge\mathrm{len}(G) + \left\lceil \log_2 \frac{1}{\epsilon} \right\rceil, \quad \text{for } \epsilon \in (0, 1).
\]

\end{proof}

Theorem~\ref{zhihuolianchangdingli} provides sufficient conditions, and Theorem~\ref{zhihuolianchangdinglijinxing} establishes their near-tightness in the worst case.

\begin{thm}[\textbf{Tightness of Theorem~\ref{zhihuolianchangdingli}}]\label{zhihuolianchangdinglijinxing}
The sufficient conditions in Theorem \ref{zhihuolianchangdingli}(i) and (ii) are nearly tight in the worst-case sense:
\begin{enumerate}[label=\textup{(\roman*)},itemsep=0pt, align=left,labelwidth=1.5em]
    \item For any \textbf{fixed chain length} $n\geq 1$, there exists a finite nilpotent group $G$ with $\mathrm{len}(G)=n$, such that if $k=\mathrm{len}(G)+\left\lceil\log_2\dfrac{1}{\epsilon}\right\rceil-2$, $\varphi_k(G)<1 - \epsilon$.
    
    \item For any \textbf{fixed rank} $r\ge 1$, there exists a finite nilpotent group $G$ with $\mathrm{rank}(G)=r$, such that if $k =\mathrm{rank}(G)+\left\lceil \log_2 \dfrac{2}{\epsilon} \right\rceil - 3$, $\varphi_k(G) < 1 - \epsilon$.
   
\end{enumerate}
\end{thm}

\begin{proof}

(i) Consider $G=(\mathbb{Z}_2)^n$. From part(ii) of Lemma~\ref{p-qungailv}, we know:
\[
\varphi_k(G)=\varphi_k\Big((\mathbb{Z}_{2})^{n}\Big)=\prod_{l=0}^{n-1}(1-\dfrac{1}{2^{k-l}}),\quad {k\ge n}
\]

Substituting $k=\mathrm{len}(G)+\left\lceil\log_2\dfrac{1}{\epsilon}\right\rceil-2$ into the above expression yields:
\begin{align*}
\varphi_k(G)&=\prod_{l=0}^{n-1}(1-\dfrac{1}{2^{k-l}})\\
&\le 1-\dfrac{1}{2^{k-n+1}}\quad \text{(bounding via the minimum factor at $l=n-1$)}\\
&=1 - \dfrac{1}{2^{\lceil \log_2(1/2\epsilon) \rceil}}\quad\text{(substituting $k=\mathrm{len}(G)+\left\lceil\log_2\dfrac{1}{\epsilon}\right\rceil-2$)}\\
&<1 - \dfrac{1}{2^{\log_2(1/2\epsilon) + 1}}\quad\text{(applying $\lceil x \rceil < x + 1$ for $x \in \mathbb{R}$)} \\
&= 1 - \dfrac{1}{2^{\log_2(1/\epsilon)}}\\
&=1-\epsilon. 
\end{align*}

Thus $k=\mathrm{len}(G)+\left\lceil\log_2\dfrac{1}{\epsilon}\right\rceil-2$ is not enough to ensure the success probability exceeds $1-\epsilon$. 

(ii) Consider $G=(\mathbb{Z}_{2})^r$. From part(ii) of Lemma~\ref{p-qungailv}, we know:
\[
\varphi_k(G)=\varphi_k\Big((\mathbb{Z}_{2})^{r}\Big)=\prod_{l=0}^{r-1}(1-\dfrac{1}{2^{k-l}}),\quad {k\ge r}
\]

Substituting $k=\mathrm{rank}(G)+\left\lceil\log_2\dfrac{2}{\epsilon}\right\rceil-3$ into the above expression yields:
\begin{align*}
\varphi_k(G)&=\prod_{l=0}^{r-1}(1-\dfrac{1}{2^{k-l}})\\
&\le (1-\dfrac{1}{2^{k-r+1}})\quad \text{(bounding via the minimum factor at $l=r-1$)}\\
&=(1-\dfrac{1}{2^{\lceil \log_2(1/2\epsilon) \rceil}})\quad\text{(substituting $k=\mathrm{rank}(G)+\left\lceil\log_2\dfrac{2}{\epsilon}\right\rceil-3$)}\\
&<1-\dfrac{1}{2^{\log_2(1/\epsilon)}}\quad\text{(applying $\lceil x \rceil < x + 1$ for $x \in \mathbb{R}$)} \\
&=1-\epsilon.
\end{align*}

Thus $k=\mathrm{rank}(G)+\left\lceil\log_2\dfrac{2}{\epsilon}\right\rceil-3$ is not enough to ensure the success probability exceeds $1-\epsilon$.
\end{proof}

\begin{remark}

\begin{itemize}
    \item For \textbf{fixed chain length}, from Eq.~(\ref{shizi89}) in Theorem~\ref{zhihuolianchangdingli}(ii), we derive the inequality  
\[ 
\varphi_k(G) \geq \varphi_k\left( (\mathbb{Z}_2)^{\mathrm{len}(G)}\right).
\]  
This implies that among all nilpotent groups with identical chain length $n$, $(\mathbb{Z}_2)^n$ is ``the most difficult'' group to generate; while the cyclic group $\mathbb{Z}_{p^n}$ is ``the easiest'' group to generate, where $p$ is the largest known prime ($2^{136279841} - 1$ up to year 2024). 

    \item For \textbf{fixed rank}, from Eqs.~(\ref{shizi67}) and (\ref{shizi68}) in Theorem~\ref{zhihuolianchangdingli}(i), we derive the inequality 
\[ 
\varphi_k(G)\ge\prod_{p_i\mid |G|}\varphi_k\Big((\mathbb{Z}_{p_i})^r\Big)=\prod_{p_i\mid|G|}\prod_{j=0}^{r-1}(1-\dfrac{1}{{p_i}^{k-j}}).
\]

This implies that among all nilpotent groups with identical rank $r$, the direct product of all $p$-groups of rank $r$ over every prime $p$ is ``the most difficult'' to generate, while the $p$-group of rank $r$ where $p$ is the largest known prime ($2^{136279841} - 1$ up to year 2024) is ``the easiest'' to generate.

    \item For a specific finite nilpotent group $G$, one can theoretically use Eq.~(\ref{shizi66}) to compute the exact lower bound of $k$ for a given $\epsilon$:
    \[ 
    \varphi_k(G)=\prod_{p_i\mid |G|}\varphi_k\left((\mathbb{Z}_{p_i})^{r_i}\right)=\prod_{p_i\mid|G|}\prod_{j=0}^{r_i-1}\left(1-\dfrac{1}{{p_i}^{k-j}}\right)\ge 1-\epsilon.
    \]

In practice, however, \textbf{closed-form solutions} for exact lower bound of $k$ are generally unavailable. They are attainable only in specific cases, such as for \textbf{$p$-groups} of rank~$r \leq 4$, where the problem of determining $k$ reduces to solving a polynomial equation of degree at most $4$ in $x = p^{-k}$, which admits solutions in radicals.

For all other cases---including higher-rank $p$-groups ($r\geq 5$) and general non-$p$-group nilpotent groups---the relevant inequality is transcendental in the variable $k$, and numerical methods provide a feasible approach.

    \item Given these computational difficulties, Theorem~\ref{zhihuolianchangdingli} provides an \textbf{efficiently computable sufficient condition} for $k$ that applies universally to any finite nilpotent group $G$. The bounds 
\[ 
k \geq \operatorname{rank}(G) + \left\lceil \log_2 \frac{2}{\epsilon} \right\rceil \quad \text{and} \quad k \geq \operatorname{len}(G) + \left\lceil \log_2 \frac{1}{\epsilon} \right\rceil 
\] 
are \textbf{nearly tight} for worst-case groups, as demonstrated in Theorem~\ref{zhihuolianchangdinglijinxing}.  

\end{itemize}

\end{remark}

\section{Application to Analysis of Quantum Algorithms}\label{Sec4}

\subsection{Overview of Hidden Subgroup Problem}

This subsection reviews the Hidden Subgroup Problem (HSP) and its abelian variant (AHSP) in quantum computing, along with the definition of the orthogonal subgroup $H^\perp$.

\begin{defi}[\textbf{Hidden Subgroup Problem (HSP)}~\cite{kitaev1995quantum,nielsen_quantum_2010}]\label{HSPdingyi}
Let \( G \) be a group and \( S \) be a finite set. Given a black-box function \( f: G \to S \), suppose there exists an unknown subgroup \( H \leq G \) such that for all \( x, y \in G \),
\begin{align*}
    f(x) = f(y) \quad \text{if and only if} \quad xH = yH,\\
\text{i.e., } x^{-1}y\in H \text{ or } y^{-1}x \in H.
\end{align*}
The goal of HSP is to identify subgroup \( H \) (or a generating set for \( H \)) by querying \( f \).
\end{defi}

\begin{remark}
$xH$ and $yH$ are both left coset. The function \( f \) is \textit{constant} on left cosets of \( H \) and \textit{distinct} across different cosets. When $G$ is a finite Abelian group, the problem is specifically referred to as the finite Abelian Hidden Subgroup Problem (AHSP).

\end{remark}

Any finite Abelian group $G$ admits a decomposition into cyclic groups. We employ the \textit{elementary divisors} form:
\[
G \cong \mathbb{Z}/N_1\mathbb{Z} \oplus \mathbb{Z}/N_2\mathbb{Z} \oplus \cdots \oplus \mathbb{Z}/N_l\mathbb{Z},
\]
where $N_i = p_i^{\alpha_i}$ are prime powers. Under this decomposition, elements $x \in G$ are represented as
an ordered $l$-tuple $(x_1, \dots, x_l)$ with $x_i \in \{0, 1, \dots, N_i-1\}$ for each $i \in \{1, 2, \dots, l\}$. We now define the orthogonal subgroup $H^{\perp}$ within this framework.

\begin{defi}[\cite{kaye2006introduction}]\label{zhengjiaoqun}
Let $G$ be a finite Abelian group. For any subgroup $H \leq G$, the subgroup $H^{\perp}$ is defined as
\[
H^{\perp} := \left\{ t \in G \mid \sum_{i=1}^l \frac{t_i s_i}{N_i} \equiv 0 \pmod{1} \text{ for all } s \in H \right\},
\]
where $x \equiv 0 \pmod{1}$ means $x \in \mathbb{Z}$.
\end{defi}

\begin{lemma}[\cite{serre1977linear, lomont2004hidden}]\label{zhengjiaotonggoudingli}
Let $G$ be a finite Abelian group and $H \leq G$ be a subgroup. Then:

\begin{enumerate}[label=\textup{(\roman*)}, itemsep=0pt,align=left,labelwidth=1.5em]
   \item $H^{\perp}\cong G/H$;
    \item $H^{\perp\perp}:=(H^{\perp})^{\perp}=H$.
\end{enumerate}

\end{lemma}

Proposition~\ref{lianchangguanxi} plays a crucial role in our analysis of the standard AHSP algorithm. The additivity of chain length allows us to express $\mathrm{len}(H^{\perp})$ in terms of $\mathrm{len}(G)$ and $\mathrm{len}(H)$. This enables
us to determine query complexity directly from $\mathrm{len}(H^{\perp})$ rather than from $\mathrm{len}(G)$. Furthermore, Proposition~\ref{zhiguanxi} demonstrates that \textit{a priori} knowledge of $\mathrm{rank}(H)$ does not reduce the minimal number of algorithm iterations.

\begin{Pp}[\textbf{Additivity of Chain Length in Abelian Groups}]\label{lianchangguanxi}
For a finite Abelian group $G$ and subgroup $H\leq G$:
\begin{equation*}
\mathrm{len}(G) = \mathrm{len}(H) + \mathrm{len}(H^{\perp}).
\end{equation*}
\end{Pp}

\begin{proof}

This result is a direct consequence of Lemma~\ref{thm:normal-subgroup-length} in Abelian setting, where all subgroups are normal. 
Lemma~\ref{zhengjiaotonggoudingli} gives $H^{\perp} \cong G/H$, and since isomorphic groups have equal length, we obtain the equality.

\end{proof}

\begin{remark}
By the additivity of chain length, $\mathrm{len}(H^{\perp})$ is determined by $\mathrm{len}(H)$ and $\mathrm{len}(G)$, which 
can significantly reduce the number of algorithm iterations.
\end{remark}

\begin{Pp}[\textbf{Subadditivity of Rank in Abelian Groups}]\label{zhiguanxi}
For a finite Abelian group $G$ and subgroup $H\leq G$:
\begin{equation*}
\mathrm{rank}(G)\le \mathrm{rank}(H)+\mathrm{rank}(H^{\perp}).
\end{equation*}
\end{Pp}

\begin{proof}

This result follows directly from Lemma~\ref{subadditivity-finite} in the Abelian setting, where all subgroups are normal. 
Lemma~\ref{zhengjiaotonggoudingli} gives $H^{\perp} \cong G/H$, and since isomorphic groups have equal rank, we obtain:
\[
\mathrm{rank}(G) \le \mathrm{rank}(H) + \mathrm{rank}(H^{\perp}).
\]

\end{proof}

\begin{remark}
The subadditivity of rank only yields the lower bound 
\[
\mathrm{rank}(H^{\perp}) \geq \mathrm{rank}(G) - \mathrm{rank}(H),
\] 
but provides no information about the upper bound or exact value of $\mathrm{rank}(H^{\perp})$. Since the number of algorithm iterations depends critically on the exact rank of $H^{\perp}$, this inequality cannot reduce the number of iterations.
\end{remark}

\subsection{Application to Analysis of Quantum Algorithm for Finite AHSP}

After the $i$-th iteration of the standard quantum algorithm for AHSP, we obtain an element $\mathbf{t}_i=(t_{i1}, t_{i2},\dots, t_{il})\in H^{\perp}$, and all elements $\mathbf{t}_i\in {H}^{\perp}$ are equally likely to be measured~\cite{kaye2006introduction}. After $k$ iterations, we can only guarantee that $\langle\mathbf{t}_1,\mathbf{t}_2,\cdots,\mathbf{t}_k\rangle\subseteq{H}^{\perp}$, rather than $\langle\mathbf{t}_1,\mathbf{t}_2,\cdots,\mathbf{t}_k\rangle={H}^{\perp}$. We then attempt to recover the hidden subgroup $H$, by solving the linear congruence system:
\[
\mathbf{W}\mathbf{x} \equiv \mathbf{0} \pmod{1} \quad \text{where } 
\mathbf{W} = \begin{pmatrix} 
\mathbf{w}_1 \\ 
\vdots \\ 
\mathbf{w}_k 
\end{pmatrix}
= \begin{pmatrix} 
t_{11}/N_1 & t_{12}/N_2 & \cdots & t_{1l}/N_l \\
\vdots & \vdots & \ddots & \vdots \\
t_{k1}/N_1 & t_{k2}/N_2 & \cdots & t_{kl}/N_l
\end{pmatrix}.
\]

The solution submodule $A = \{\mathbf{x} \in G \mid \mathbf{W}\mathbf{x} \equiv \mathbf{0} \pmod{1}\}$ contains $H^{\perp\perp}$, where $H^{\perp\perp}:=(H^{\perp})^{\perp}$ denotes the orthogonal subgroup of $H^\perp$.

Theorem~\ref{suanfachaxunfuzadu}, which follows from Theorem~\ref{zhihuolianchangdingli}, guarantees that $\langle\mathbf{t}_1,\mathbf{t}_2,\cdots,\mathbf{t}_k\rangle={H}^{\perp}$ with probability exceeding $1-\epsilon$ after either \(\mathrm{rank}(G) + \lceil \log_2(2/\epsilon)\rceil\) or \(\mathrm{len}(G) - \mathrm{len}(H) + \lceil\log_2(1/\epsilon)\rceil\) iterations. The former bound achieves an exponential improvement in the $\epsilon$-dependence over the prior result $\lfloor 4/\epsilon \rfloor \mathrm{rank}(G)$~\cite{koiran2007quantum}, while the latter improves upon $\lceil\log_2|G|+\log_2(1/\epsilon)+2\rceil$~\cite{lomont2004hidden}.

If the chain length of the subgroup $H$ is unavailable, we replace the iteration bound $k \ge\mathrm{len}(G)-\mathrm{len}(H)+\lceil\log_2\dfrac{1}{\epsilon}\rceil$ with $k \ge\mathrm{len}(G)+\lceil\log_2\dfrac{1}{\epsilon}\rceil$.

\begin{thm}\label{suanfachaxunfuzadu}
In the quantum algorithm for the finite AHSP, to achieve a success probability of at least \(1 - \epsilon\), it is sufficient to perform  
\[
k\ge \mathrm{rank}(G) + \left\lceil \log_2 \dfrac{2}{\epsilon} \right\rceil \quad \text{or} \quad k \ge \mathrm{len}(G) - \mathrm{len}(H) + \left\lceil \log_2 \dfrac{1}{\epsilon} \right\rceil
\]  
iterations.
\end{thm}

\begin{proof}

This follows directly from Theorem~\ref{zhihuolianchangdingli}. Note that ensuring \(\Pr(A = H) \ge 1 - \epsilon\) is equivalent to ensuring \(\Pr(\langle \mathbf{t}_1, \mathbf{t}_2, \dots, \mathbf{t}_k \rangle = H^\perp) = \varphi_k(H^\perp) \ge 1 - \epsilon\). By Theorem~\ref{zhihuolianchangdingli}, this probability holds if we choose:
\[
k \ge \mathrm{rank}(H^\perp) + \left\lceil \log_2 \frac{2}{\epsilon} \right\rceil \quad \text{or} \quad k \ge \mathrm{len}(H^\perp) + \left\lceil \log_2 \frac{1}{\epsilon} \right\rceil.
\]
Since $\mathrm{rank}(H^\perp)\le\mathrm{rank}(G)$ and $\mathrm{len}(H^\perp) = \mathrm{len}(G) - \mathrm{len}(H)$, 
we obtain the required bounds:

\[
k \ge \mathrm{rank}(G) + \left\lceil \log_2 \frac{2}{\epsilon} \right\rceil \quad \text{or} \quad k \ge \mathrm{len}(G) - \mathrm{len}(H) + \left\lceil \log_2 \frac{1}{\epsilon} \right\rceil.
\]
 
\end{proof}

\subsection{Application to Analysis of Regev's Factoring Algorithm}

Our result contributes to establishing the requisite number of iterations for Regev's quantum factoring algorithm~\cite{regev2025efficient}. Regev's original analysis in Corollary~4.2 shows that for a finite Abelian group $G$, $\mathrm{rank}(G) + 4$ uniformly random elements generate $G$ with probability at least $1/2$. 

Since finite Abelian groups are nilpotent (of class 1), we apply Theorem~\ref{zhihuolianchangdingli}(i) with $\epsilon=1/2$ to obtain an improved bound:

\begin{thm}\label{gaijinregevjie}
Suppose $G$ is a finite Abelian group. Then $\mathrm{rank}(G) + 2$ uniformly random elements of $G$ generate $G$ with probability at least $1/2$.
\end{thm}

As a direct consequence, for an $n$-bit integers, the quantum circuit repetition count in Regev's factoring algorithm is optimized from $\sqrt{n} + 4$ to $\sqrt{n} + 2$, while maintaining the same success probability.

\section{Conclusions}\label{Sec5}
This paper establishes a quantitative sampling theorem for finite nilpotent groups. We prove that $k \ge \operatorname{rank}(G) + \lceil \log_2(2/\epsilon) \rceil$ or $k \ge \operatorname{len}(G) + \lceil \log_2(1/\epsilon) \rceil$ random elements suffice to generate $G$ with probability at least $1 - \epsilon$, and demonstrate that these bounds are nearly tight. This significantly sharpens the previously known requirement of $k \ge \lceil \log_2 |G| + \log_2(1/\epsilon) \rceil + 2$.

Our theorem provides a foundational tool for the analysis of probabilistic algorithms, leading to concrete improvements in quantum computation:

\begin{itemize}
    \item \textbf{Finite AHSP Algorithm:} We determine the iteration count for achieving success probability $1-\epsilon$ in the standard quantum algorithm to be either $\mathrm{rank}(G) + \lceil \log_2 (2/\epsilon) \rceil$ or $\mathrm{len}(G) - \mathrm{len}(H) + \lceil \log_2 (1/\epsilon) \rceil$. The former offers an exponential improvement in $\epsilon$-dependence over the prior bound $\lfloor 4/\epsilon \rfloor \mathrm{rank}(G)$~\cite{koiran2007quantum}, while the latter improves upon $\lceil \log_2 |G| + \log_2 (1/\epsilon) + 2 \rceil$~\cite{lomont2004hidden}.

    \item \textbf{Regev's Factoring Algorithm:} Our result optimizes the iteration counts in Regev's recent quantum factoring algorithm~\cite{regev2025efficient}, thereby directly reducing the quantum circuit repetition count from $\sqrt{n} + 4$ to $\sqrt{n} + 2$ while maintaining identical success probability.
\end{itemize}

The probabilistic bounds developed here provide a foundational building block for future quantum and classical randomized algorithms, paving the way for applications beyond this work. We will leverage the present results in the design of algorithms in forthcoming work.

\section*{Acknowledgements}
This work is supported in part by the National Natural Science Foundation of China (Nos. 61876195, 11801579), and the Natural Science Foundation of Guangdong Province of China (No. 2017B030311011).

\bibliographystyle{unsrt}
\bibliography{ref}


\end{document}